\documentclass[conference]{IEEEtran}
\IEEEoverridecommandlockouts
\usepackage{cite}
\usepackage{graphicx}
\usepackage{textcomp}
\usepackage{xcolor}
\usepackage{array}
\usepackage{tabularx}
\usepackage{makecell}
\usepackage{array}
\usepackage{tabularx}
\usepackage{makecell}
\usepackage{multirow}
\newcolumntype{Y}{>{\centering\arraybackslash}X}

\usepackage{amsmath,amssymb,amsfonts}
\usepackage{amsthm}
\usepackage[colorlinks=true, linkcolor=blue, citecolor=blue, urlcolor=blue]{hyperref}
\usepackage{algorithm}
\usepackage{algpseudocode}
\usepackage{amsmath, amssymb}
\usepackage{algorithm}
\usepackage{enumitem}
\usepackage{booktabs}
\usepackage{algpseudocode}

\newtheorem{assumption}{Assumption}
\newtheorem{definition}{Definition}

\newtheorem{proposition}{Proposition}

\def\BibTeX{{\rm B\kern-.05em{\sc i\kern-.025em b}\kern-.08em
    T\kern-.1667em\lower.7ex\hbox{E}\kern-.125emX}}
\begin{document}
\bstctlcite{IEEEexample:BSTcontrol}
\title{\LARGE Prior-Guided Movable Antenna Control for Agile Multi-Path Sensing
}

\vspace{-3mm}
\author{
Jaehong~Kim$^{*}$, Jihong~Park$^{\dagger}$, 
Changsheng You$^{\ddagger}$, and Seung\mbox{-}Woo~Ko$^{*}$\\
\small $^{*}$Inha University, Korea, e-mail: kimjaehong@inha.edu, swko@inha.ac.kr\\
\small $^{\dagger}$Singapore University of Technology and Design, Singapore, e-mail: 
jihong\_park@sutd.edu.sg \\
\small $^{\ddagger}$ Southern University of Science and Technology, China, e-mail: youcs@sustech.edu.cn
\vspace{-5mm}
}

\maketitle

\begin{abstract}
Multi-path sensing, which aims to extract the geometric attributes of multiple propagation paths, is expected to be a key functionality of 6G. A \emph{movable antenna} (MA) can enable this functionality by creating a synthetic aperture through sequential mechanical motion. However, existing MA-based sensing methods typically rely on exhaustive scanning over the entire movable plate, resulting in significant control overhead and sensing latency, which limits their practicality for agile sensing. To address this challenge, this paper develops a prior-guided agile multi-path sensing framework that leverages weak prior \emph{angle-of-arrival} (AoA) statistics as side information. The proposed framework comprises two steps. First, the movable plate's three-dimensional orientation is optimized only once to maximize path visibility while preserving path discriminability, both induced from Fisher information analysis. Second, only two predetermined linear MA scans are made on the tilted plate to estimate the elevation and azimuth AoAs from the resulting sequence of received signals. By incorporating the prior AoA statistics, a \emph{maximum a posteriori} (MAP)-based AoA estimation algorithm is developed. With only one orientation control and two linear scans, the proposed framework enables agile multi-path sensing with significantly reduced control overhead and latency, while achieving AoA estimation accuracy approaching that of the single-path benchmark. 
\end{abstract}


\section{Introduction}
6G is expected to exploit radio signals not only for communication but also for sensing \cite{FLiu2022}, since the intrinsic multi-path nature of radio waves enables the detection of objects and the inference of environmental geometry even under \emph{Non-Line-of-Sight}~(NLoS) conditions \cite{Lotti2023, Yang2025}. We refer to the process of extracting the geometric attributes of these propagation paths, specifically their \emph{Angle-of-Arrivals}~(AoAs) and \emph{time-of-arrival}~(ToAs), without performing full channel estimation, as \emph{multi-path sensing}. This work focuses on AoA estimation, while extension to other components is left for future work. 

A promising enabler for multi-path sensing is the \emph{movable antenna}~(MA), which forms a synthetic aperture by mechanically relocating the antenna over a confined spatial region, referred to as the movable plate \cite{Zhou2025, Wu2026}. The resulting signals, collected from multiple locations and embedding rich spatial information from diverse directional perspectives, enable the resolution of multi-path propagation and the extraction of its constituent components, even in challenging scenarios where multiple paths arrive from similar directions and are therefore difficult to distinguish.

Several works on MA-based sensing have been proposed in the literature by applying optimization techniques into the MA framework. For example, the technique of compressed sensing is employed in \cite{MA2023}, where multi-path components, including AoAs and ToAs, are estimated by alternatively scanning the movable plates at the transmitter and receiver. In \cite{ZXiao2024}, the technique of \emph{orthogonal matching pursuit}~(OMP) is applied to multi-path sensing by correlating the received pilots with a position-dependent dictionary constructed from the MA locations and progressively canceling already detected path parameters from the received signals.
This idea is further extended to wideband scenario in \cite{SCao2025}, which develops a \emph{simultaneous OMP}~(SOMP)-based framework by jointly exploiting the common sparsity across subcarriers. Despite their effectiveness, these approaches rely on scanning the entire movable plate, which incurs substantial mechanical overhead and excessive sensing latency. Consequently, they are not well suited to agile sensing, another key requirement of 6G beyond estimation accuracy \cite{Chen2022}. 

To realize agile multi-path sensing, this work exploits coarse AoA statistics inferred from surrounding environments, such as nearby static buildings and obstacles, which can be obtained through various technique, including large multimodal model \cite{Kim2026} and digital twins \cite{Abouamer2025}. We treat such prior statistics as weak but useful side information and develop an agile MA control comprising two steps. First, the movable plate's \emph{three-dimensional}~(3D) orientation is optimized only once so as to maximize the visibility of all propagation paths while maintaining sufficient path discriminability, both characterized from a Fisher information perspective. Second, two predetermined linear MA scans are performed on the tilted plate, from which the elevation and azimuth AoAs are estimated from the resulting sequence of received signals. By incorporating the prior AoA statistics, we develop a \emph{maximum a posteriori}~(MAP)-based AoA estimation algorithm, whose effectiveness is verified by simulation that the AoA estimation accuracy approaches that of the single-path benchmark. Notably, the proposed framework 
achieves accurate multi-path AoA resolution using only one orientation optimization and two linear MA scans, thereby significantly reducing control overhead and sensing latency, which is the main contribution of this work.

\section{System Model}
This section presents the system model for the considered multi-path sensing scenario, including the geometry and signal model, and formulates the problem addressed in this work.

\subsection{Scenario Description}\label{AA}
Consider a sensing network comprising a transmitter and a receiver equipped with a single rigid antenna and a single MA, respectively.  The MA is mounted on a mechanically adjustable \emph{movable plate} and can move over the plate, thereby enabling the formation of a synthetic aperture. We define an initial \emph{three-dimensional} Cartesian coordinate system \((\mathsf{X}^{(0)},\mathsf{Y}^{(0)},\mathsf{Z}^{(0)})\) with the receiver's reference point at the origin, \(\boldsymbol{p}_0=[0,0,0]^\top\), as shown in Fig.~\ref{Fig:MA_coordinate}\,(a).
\begin{figure}
  \centering
  \includegraphics[width=0.9\linewidth]{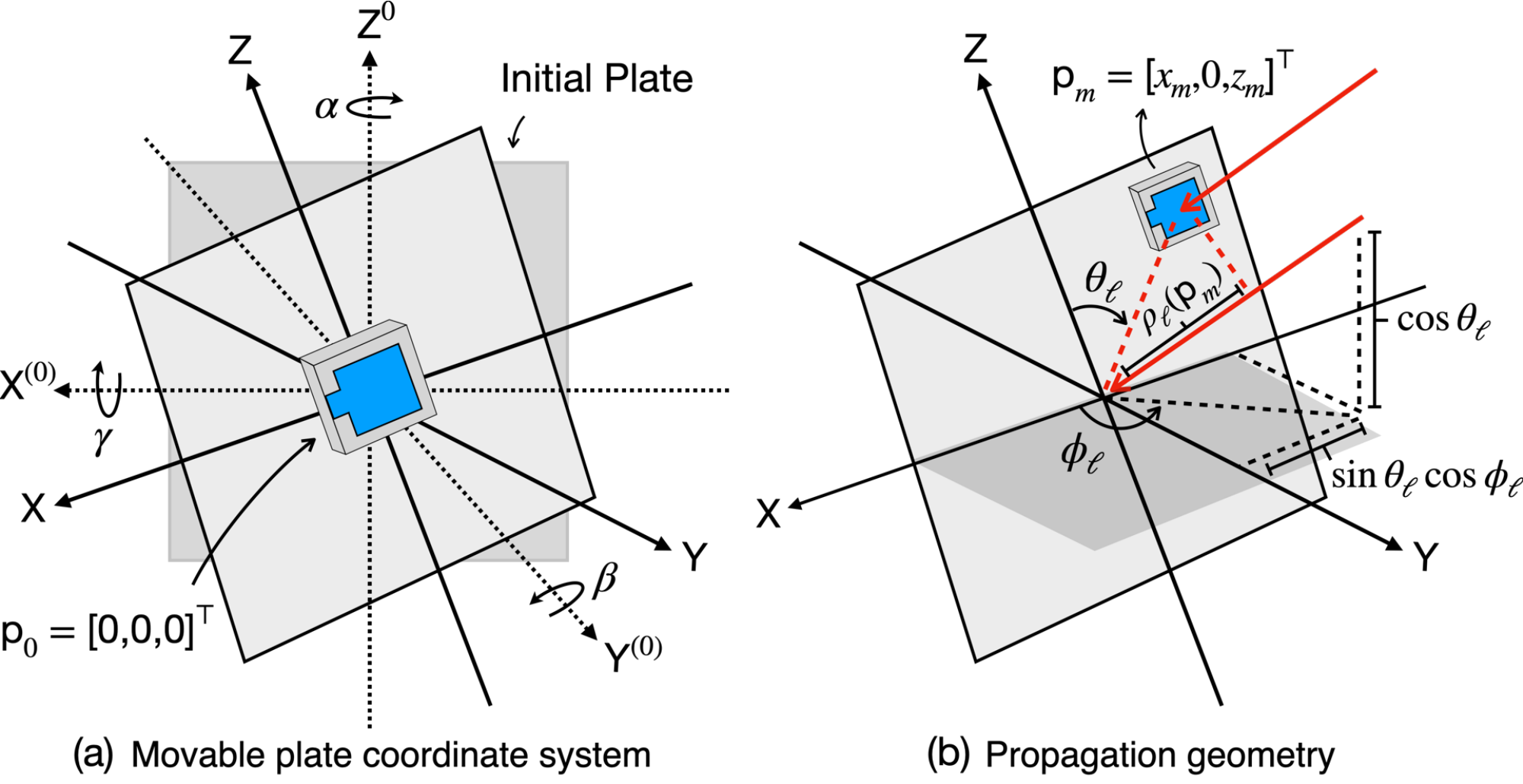}
  \vspace{-3mm}
  \caption{The geometries of movable plate and MA. (a) Coordinate system before and after orientating the movable plate. (b) Additional propagation distance when MA moves from \(\boldsymbol{p}_0\) to \(\boldsymbol{p}_m\).}
  \vspace{-3mm}
  \label{Fig:MA_coordinate}
\end{figure}
We consider $L$ NLoS \emph{signal paths} (SPs), indexed by $\ell\in\{1,\cdots,L\}$. Each SP $\ell$ can be characterized by the elevation AoA $\theta_{\ell}^{(0)}$ and the azimuth AoA $\phi_{\ell}^{(0)}$, defined as
\begin{align}
\theta_{\ell}^{(0)}
= \cos^{-1}\!\left(\frac{[\boldsymbol{b}_\ell]_3}{\lVert \boldsymbol{b}_\ell \rVert}\right),
\quad
\phi_{\ell}^{(0)}
= {\tan}^{-1}\left(\frac{[\boldsymbol{b}_\ell]_2}{[\boldsymbol{b}_\ell]_1}\right),
\label{eq:angles}
\end{align}
where $\boldsymbol{b}_{\ell}\in\mathbb{R}^{3\times 1}$ is the last reflection point before the arrival of SP $\ell$, and $[\cdot]_i$ denotes the $i$-th element of a vector. The corresponding unit direction vector of SP $\ell$, denoted by $\boldsymbol{a}_{\ell}^{(0)}$, is given as 
\begin{align}\label{eq: Initial AoA direction}
\boldsymbol{a}_\ell^{(0)} = [\sin \theta_\ell^{(0)} \cos \phi_\ell^{(0)}, \sin \theta_\ell^{(0)} \sin \phi_\ell^{(0)}, \cos \theta_\ell^{(0)}]^\top.
\end{align}
Throughout this work, we aim to precisely estimate the elevation and azimuth AoAs of all SPs, namely, $\{\theta_{\ell}^{(0)}, \phi_{\ell}^{(0)}\}_{\ell=1}^{L}$.

\subsection{Signal and Channel Models}\label{Sec: Signal Model}
\textit{1) Transmit signal:}
Consider a downlink OFDM system in which the transmitter sends
$\boldsymbol{x}(t)=[x^{(1)}(t),\cdots,x^{(K)}(t)]^{\mathrm{T}}\in\mathbb{C}^{K\times1}$
over $K$ subcarriers. With frequency spacing $\delta$ and frame duration $T=1/\delta$, the $k$-th subcarrier signal, denoted by $x^{(k)}(t)$, is given as
\begin{align}
x^{(k)}(t)=\sqrt{\frac{P}{K}}e^{j2\pi f_k t}, \quad 0\le t<T,
\end{align}
where $P$ is the total transmit power and $f_k$ is the frequency of the $k$-th subcarrier for $k=1,\dots,K$.

 \textit{2) Channel via MA:}
The channel depends on the movable plate orientation and the MA position within it. We describe these two factors as follows. 
 
First, the movable plate is tilted by angles $\alpha$, $\beta$, and $\gamma$ about the $\mathsf{Z}^{(0)}$, $\mathsf{Y}^{(0)}$, and $\mathsf{X}^{(0)}$-axes, respectively. We then define a rotated coordinate system $(\mathsf{X},\mathsf{Y},\mathsf{Z})$ such that the tilted plate lies on the $\mathsf{X}$-$\mathsf{Z}$ plane and its outward normal aligns with $\mathsf{Y}$ (see Fig.~\ref{Fig:MA_coordinate}\,(a)). In this frame, the arrival direction of SP $\ell$ is 
\begin{align}\label{eq: revised AoA direction}
    \boldsymbol{a}_\ell=\boldsymbol{R}^{\top}(\alpha,\beta,\gamma)\boldsymbol{a}^{(0)}_\ell,
\end{align}
where $\boldsymbol{R}(\alpha,\beta,\gamma)$ is the Euler rotation matrix \cite{Diebel2006}. The corresponding local AoAs, denoted by \(\theta_\ell\) and \(\phi_\ell\), can be characterized by replacing \(\boldsymbol a_\ell^{(0)}\) with \(\boldsymbol a_\ell\) in \eqref{eq:angles}, whose detailed form is omitted due to the page limit. 

Second, to account for the MA location, let the MA move from \(\boldsymbol{p}_0\) to \(\boldsymbol{p}_m\), where \(\boldsymbol{p}_m=[x_m,0,z_m]^{\mathrm{T}}\) denotes the MA position in the rotated coordinate system (see Fig.~\ref{Fig:MA_coordinate}\,(b)). Assuming a sufficiently long channel coherence time, the channel parameters remain unchanged during the MA movement. The channel vector of the $\ell$-th SP at location $\boldsymbol{p}_m$ is defined as
$\boldsymbol{h}_\ell(\boldsymbol{p}_m)=[h_\ell^{(1)}(\boldsymbol{p}_m),\dots,h_\ell^{(K)}(\boldsymbol{p}_m)]\in\mathbb{C}^{1\times K}$.
The element $h_{\ell}^{(k)}(\boldsymbol{p}_m)$ denotes the channel response of SP $\ell$ on the $k$-th subcarrier, given by  
\begin{align}
h_{\ell}^{(k)}(\boldsymbol{p}_m)
= \upsilon_\ell \exp\!\left(
-j2\pi f_k\left(\tau_\ell+\frac{\rho_{\ell}(\boldsymbol{p}_m)}{c}\right)
\right),
\end{align}
where $\upsilon_\ell$ is the path attenuation factor and $c\approx 3\times10^8$ m/s is the speed of light. Here, \(\rho_\ell(\boldsymbol p_m)\) is the additional propagation distance, given by
\begin{align}\label{eq: propagation dist}
\rho_\ell(\boldsymbol{p}_m)
&= \boldsymbol{a}_\ell^\top \boldsymbol{p}_m = x_m \sin \theta_\ell \cos \phi_\ell + z_m \cos \theta_\ell.
\end{align}
Consequently, the overall channel at location $\boldsymbol{p}_m$ on the $k$-th subcarrier is
\begin{align}\label{eq: Overall channel}
h^{(k)}(\boldsymbol{p}_m)
= \sum_{\ell=1}^L \upsilon_\ell
\exp\!\left(
-j2\pi f_k\left(\tau_\ell+\frac{\rho_\ell(\boldsymbol{p}_m)}{c}\right)
\right).
\end{align}

 \textit{3) Received signal:}
 The received signal on the $k$-th subcarrier at time $t$, when the MA is at position $\boldsymbol{p}_m$, denoted by $y^{(k)}(t;\boldsymbol{p}_m)$, is given as 
\begin{align}
{y}^{(k)}(t;\boldsymbol{p}_m)\!=\! \sum_{\ell =1}^L \upsilon_\ell 
   e^{-j2\pi f_k\!\left(
      \!\tau_\ell 
      + \frac{\rho_\ell(\boldsymbol{p}_m)}{c}
     \right)\!} 
   x^{(k)}(t) \!+\! w^{(k)}(t),
\end{align}
where \(w^{(k)}(t)\in\mathbb{C}\) is the thermal noise on the \(k\)-th subcarrier, modeled as independent \(\mathcal{CN}(0,N_0)\). The demodulated received signal vector at position \(\boldsymbol p_m\) is denoted by \(\boldsymbol s(\boldsymbol p_m)=[s^{(1)}(\boldsymbol p_m),\dots,s^{(K)}(\boldsymbol p_m)]^\top\), obtained as
\begin{align}
{s}^{(k)}(\boldsymbol{p}_m) = \frac{P}{K} \left ( \sum_{\ell =1}^L \upsilon_\ell e^{-j2\pi f_k( \tau_\ell + {  \frac{\rho_\ell(\boldsymbol{p}_m)}{c} })} \right ) + \bar{w}^{(k)}, \label{eq:received_signal}
\end{align}
where $\bar{w}^{(k)}$ is the demodulated noise component following $\mathcal{CN}(0, PN_0/K)$.

\subsection{Problem Description}\label{sec:problem_description}
Due to the mechanical overhead associated with plate actuation and MA motion, frequent control of the movable plate and the MA is impractical. 
We thus restrict the control procedure as follows. First, the orientation of the movable plate, parametrized by \((\alpha,\beta,\gamma)\), is adjusted once before signal reception. Second, after the plate orientation is fixed, the MA moves linearly along the $\mathsf{X}$- and $\mathsf{Z}$-axes on the tilted plate with constant spacing. During these movements, the MA receives and demodulates a sequence of signals. 

Based on this single orientation control and the two predetermined MA scans, our objective is to accurately estimate the AoAs \(\{(\theta_\ell^{(0)},\phi_\ell^{(0)})\}_{\ell=1}^L\) of all SPs. 
To support both the orientation design and the subsequent AoA estimation, we make the following assumption:
\begin{assumption}[Prior AoA Statistics]\label{Assumption of AoA Prior} \emph{The last reflection point $\boldsymbol{b}_{\ell}$ of SP $\ell$ typically lies on the surfaces of surrounding objects 
(e.g., walls and pillars), which are known in advance from a coarse environment model as  exemplified in Table~\ref{tab:aoa_prior_stats}. Consequently, the induced AoA uncertainty can be approximated by a weakly informative Gaussian prior:
\begin{align}
\label{eq:aoa_statistics}
\theta_\ell^{(0)}\sim \mathcal{N}(\mu_{\ell}, \sigma_{\ell}^2), \quad \phi_\ell^{(0)}\sim\mathcal{N}(\xi_\ell,\varsigma_\ell^2), \quad \theta_\ell^{(0)}\bot \phi_\ell^{(0)},
\end{align}
where the corresponding mean and variance pairs, say $\{\mu_{\ell}, \sigma^2_{\ell}\}$ and $\{\xi_\ell, \varsigma^2_{\ell}\}$, are known a priori.} 
\end{assumption} 
The weak prior statistics in Assumption~\ref{Assumption of AoA Prior} play a dual role: they guide the optimal plate-orientation design and assist the subsequent AoA estimation from the received signals, as detailed in the following section. 


\section[Movable Plate Orientation Control for AoA Estimation]
{Movable Plate Orientation Control}
\label{Sec:Optimization}
This section aims to derive the optimal orientation of the movable plate to enable accurate estimation of each SP's AoA. We first analyze its effect on AoA estimation from a Fisher information perspective, which guides us to formulate the optimal plate control problem in a tractable form next. 

\subsection{Fisher Information Analysis}
To derive the \emph{Fisher information matrix} (FIM) for the AoA estimation problem, let the unknown AoA parameter vector be formed by stacking the AoAs of all $L$ SPs as $\boldsymbol{\psi} = [\boldsymbol{\psi}_{1}^\top, \dots, \boldsymbol{\psi}_{L}^\top]^\top \in \mathbb{R}^{2L \times 1}$, where $\boldsymbol{\psi}_{\ell}=[\theta^{(0)}_{\ell}, \phi_{\ell}^{(0)}]^\top$ collects the elevation and azimuth angles of SP $\ell$ specified in \eqref{eq:angles}. We consider that the MA moves along a linear trajectory and collects received signals at \(M\) positions with spacing \(d\), denoted by \(\{\boldsymbol p_m\}_{m=1}^M\). The corresponding observation vector is formed as
\begin{align}
\boldsymbol s
=
[\boldsymbol s(\boldsymbol p_1)^\top,\boldsymbol s(\boldsymbol p_2)^\top,\dots,\boldsymbol s(\boldsymbol p_M)^\top]^\top
\in\mathbb C^{MK\times1},
\end{align}
where \(\boldsymbol s(\boldsymbol p_m)\) is the received signal vector at position \(\boldsymbol p_m\) defined in \eqref{eq:received_signal}.

Recalling the Gaussian model in \eqref{eq:received_signal}, the resultant FIM $\boldsymbol{I}(\boldsymbol{\psi})\in\mathbb{R}^{2L\times 2L}$ can be expressed as
\begin{align}\label{eq: Gamma}
    \boldsymbol{I}(\boldsymbol{\psi})
    = \frac{2K}{PN_0} \Re \left\{
    \begin{bmatrix}
    \boldsymbol{\Gamma}_{1,1} & \boldsymbol{\Gamma}_{1,2} & \dots & \boldsymbol{\Gamma}_{1,L} \\
    \boldsymbol{\Gamma}_{2,1} & \boldsymbol{\Gamma}_{2,2} & \dots & \boldsymbol{\Gamma}_{2,L} \\
    \vdots & \vdots & \ddots & \vdots \\
    \boldsymbol{\Gamma}_{L,1} & \boldsymbol{\Gamma}_{L,2} & \dots & \boldsymbol{\Gamma}_{L,L}
    \end{bmatrix}
    \right\},
\end{align}
where the $(\ell,u)$-th block $\boldsymbol{\Gamma}_{\ell,u} \in \mathbb{C}^{2 \times 2}$,  $\ell, u \in \{1, \dots, L\}$, is given by 
\begin{align}
\boldsymbol{\Gamma}_{\ell,u} &=  \left(\frac{\partial \boldsymbol{\mu}(\boldsymbol{\psi})}{\partial \boldsymbol{\psi}_\ell}\right)^H \left(\frac{\partial \boldsymbol{\mu}(\boldsymbol{\psi})}{\partial \boldsymbol{\psi}_u}\right)\nonumber\\
&=
\begin{cases}
{c}_\ell \sum_{m=1}^{M} \boldsymbol{g}_{\ell}(\boldsymbol{p}_m)\boldsymbol{g}_{\ell}(\boldsymbol{p}_m)^{\top}, & \ell=u,\\
\sum_{m=1}^{M} {\kappa}_{\ell,u}(\boldsymbol{p}_m)\boldsymbol{g}_{\ell}(\boldsymbol{p}_m)\boldsymbol{g}_{u}(\boldsymbol{p}_m)^{\top}& \ell\neq u,
\end{cases},
\end{align}
where \(\boldsymbol{\mu}(\boldsymbol{\psi})\) denotes the noiseless mean signal vector, and the AoA-gradient vector $\boldsymbol{g}_{\ell}(\boldsymbol{p}_m)\triangleq
\left[\frac{\partial \rho_\ell(\boldsymbol{p}_m)}{\partial \theta_\ell^{(0)}} 
,\frac{\partial \rho_\ell(\boldsymbol{p}_m)}{\partial \phi_\ell^{(0)}}\right]^\top$ represents the sensitivity of  $\rho_\ell(\boldsymbol{p}_m)$ in \eqref{eq: propagation dist} with respect to $\boldsymbol{\psi}_\ell$.
The diagonal-block scaling factor ${c}_\ell \triangleq |\upsilon_\ell|^2 \sum_{k=1}^K \left( \frac{2\pi f_k}{c}\frac{P}{K} \right)^2$ is independent of the MA's position $\boldsymbol{p}_m$. On the other hand, the off-diagonal scaling factor ${\kappa}_{\ell,u}(\boldsymbol{p}_m)$ depends on $\boldsymbol{p}_m$ and is given by
\begin{align}\label{eq: kappa}
\kappa_{\ell,u}(\boldsymbol{p}_m) \triangleq& \ \upsilon_\ell^* \upsilon_u \left( \frac{P}{K} \right)^2\!\! \sum_{k=1}^K \!\!\left(\frac{2\pi f_k}{c} \right)^2 \!\!\! 
e^{j{2\pi f_k \Delta_{\ell,u}(\boldsymbol{p}_m)}},
\end{align}
where 
\begin{align}\label{eq: Delta}
\Delta_{\ell,u}(\boldsymbol{p}_m)=\tau_\ell - \tau_u
+ \frac{\rho_\ell(\boldsymbol{p}_m) - \rho_u(\boldsymbol{p}_m)}{c},
\end{align}
which represents phase difference between two SPs $\ell$ and $u$. 

The estimation error covariance is lower-bounded by the inverse of the FIM. In view of the block structure of \(\boldsymbol{I}(\boldsymbol{\psi})\) specified in \eqref{eq: Gamma}, two strategies can be considered for increasing \(\det(\boldsymbol{I}(\boldsymbol{\psi}))\): (i) strengthening the diagonal blocks and/or (ii) suppressing the off-diagonal coupling blocks. However, increasing the diagonal blocks \(\boldsymbol\Gamma_{\ell,\ell}\) does not necessarily increase \(\mathsf{det}(\boldsymbol{I}(\boldsymbol{\psi}))\). This is because the AoA-gradient vectors \(\boldsymbol{g}_\ell(\boldsymbol{p}_m)\) that strengthen the diagonal blocks also contribute to the off-diagonal blocks, so the gain from \(\boldsymbol\Gamma_{\ell,\ell}\) may be offset by inter-parameter coupling.

\begin{table}[t]
\centering
\caption{Experimental Validation of Assumption \ref{Assumption of AoA Prior}: AoA Statistics in Regions of Interest using Sionna Ray Tracing. The detailed experiments are omitted due to the page limit.}
\vspace{-1mm}
\label{tab:aoa_prior_stats}
\setlength{\tabcolsep}{1.5pt}
\renewcommand{\arraystretch}{1.1}
\begin{tabular}{c|c|cc|cc|c}
\hline
\multirow{2}{*}{\shortstack{Region of\\Interest}} 
& \multirow{2}{*}{\shortstack{Signal\\Path}}
& $\mu_{\ell}$ & $\sigma_{\ell}$ & $\xi_{\ell}$ & $\varsigma_{\ell}$
& \multirow{2}{*}{\shortstack{Correlation\\Coefficient}} \\
\cline{3-6}
& & \multicolumn{2}{c|}{Elevation} & \multicolumn{2}{c|}{Azimuth} & \\
\hline
\multirow{4}{*}{\shortstack{Arc de Triomphe\\(Paris)}}
& SP 1 & $98^\circ$  & $1.41^\circ$ & $16^\circ$  & $3.46^\circ$ & 0.09 \\
& SP 2 & $97^\circ$  & $1.41^\circ$ & $50^\circ$  & $6.63^\circ$ & 0.06 \\
& SP 3 & $95^\circ$  & $1.00^\circ$ & $85^\circ$  & $6.78^\circ$ & 0.02 \\
& SP 4 & $96^\circ$  & $1.00^\circ$ & $116^\circ$ & $3.61^\circ$ & 0.01 \\
\hline
\multirow{3}{*}{\shortstack{Florence Duomo\\(Florence)}}
& SP 1 & $122^\circ$ & $2.45^\circ$ & $4^\circ$   & $2.24^\circ$ & 0.06 \\
& SP 2 & $129^\circ$ & $1.73^\circ$ & $22^\circ$  & $7.35^\circ$ & 0.12 \\
& SP 3 & $102^\circ$ & $3.32^\circ$ & $38^\circ$  & $1.41^\circ$ & -0.07 \\
\hline
\end{tabular}
\vspace{-5mm}
\end{table}

Therefore, we focus on suppressing the off-diagonal blocks \(\boldsymbol\Gamma_{\ell,u}\), which are proportional to \(\kappa_{\ell,u}(\boldsymbol{p}_m)\) in \eqref{eq: kappa} and decrease as \( |\rho_\ell(\boldsymbol{p}_m)-\rho_u(\boldsymbol{p}_m)| \) in \eqref{eq: Delta} increases due to phase misalignment across subcarriers. Motivated by this observation, the next subsection optimizes the movable plate orientation \(\boldsymbol{\varphi}\triangleq[\alpha,\beta,\gamma]\) to enlarge \(|\rho_\ell(\boldsymbol p_m)-\rho_u(\boldsymbol p_m)|\) for all SP pairs under the stochastic AoA prior in Assumption~\ref{Assumption of AoA Prior}.

\subsection{Movable Plate Orientation Control}
Motivated by the above discussion, we next develop a practical plate-orientation control method based on two predetermined MA trajectories along the \(\mathsf{X}\)- and \(\mathsf{Z}\)-axes with step size \(d\), namely,
\begin{align}
\boldsymbol{p}_m^{(1)}=(m-1)d \mathbf{i},\quad 
\boldsymbol{p}_m^{(2)}=(m-1)d \mathbf{k},
\end{align}
where \(\mathbf{i}=[1,0,0]^\top\) and \(\mathbf{k}=[0,0,1]^\top\) are standard basis vectors along the $\mathsf{X}$- and $\mathsf{Z}$-axes, respectively. Collectively, these two scan lines traverse the plate in orthogonal directions, providing broad spatial coverage with respect to the AoA.

Given an arbitrary plate orientation \(\boldsymbol{\varphi}=[\alpha,\beta,\gamma]\), the propagation distance variations of SP \(\ell\) induced by the MA movements $\boldsymbol{p}_m^{(1)}$ and $\boldsymbol{p}_m^{(2)}$ are 
\begin{align}\label{eq:scan}
\rho_{\ell}(\boldsymbol{p}_m^{(1)})&=(m-1)d\,\rho_{\ell}(\mathbf{i})
=(m-1)d\,[\boldsymbol{R}^{\top}(\boldsymbol{\varphi})\boldsymbol{a}^{(0)}_\ell]_1,\nonumber\\
\rho_{\ell}(\boldsymbol{p}_m^{(2)})&=(m-1)d\,\rho_{\ell}(\mathbf{k})
=(m-1)d\,[\boldsymbol{R}^{\top}(\boldsymbol{\varphi})\boldsymbol{a}^{(0)}_\ell]_3,
\end{align}
both of which scale linearly with the travel distance $(m-1)d$.

On the other hand, the remaining terms $\rho_{\ell}(\mathbf{i})$ and $\rho_{\ell}(\mathbf{k})$ are the scalar projections of the rotated arrival vector $\boldsymbol{R}^{\top}(\boldsymbol{\varphi})\boldsymbol{a}^{(0)}_\ell$ onto the unit directions $\mathbf{i}$ ($\mathsf{X}$-axis) and $\mathbf{k}$ ($\mathsf{Z}$-axis), respectively. 
Equivalently, 
\begin{align}
\rho_{\ell}(\mathbf{i})=\mathbf{i}^\top\boldsymbol{R}^{\top}(\boldsymbol{\varphi})\boldsymbol{a}^{(0)}_\ell,\quad
\rho_{\ell}(\mathbf{k})=\mathbf{k}^\top\boldsymbol{R}^{\top}(\boldsymbol{\varphi})\boldsymbol{a}^{(0)}_\ell.
\end{align}
Since $\boldsymbol{\varphi}$ affects the propagation-distance variation through these two projections, optimizing $\boldsymbol{\varphi}$ amounts to shaping $\rho_{\ell}(\mathbf{i})$ and $\rho_{\ell}(\mathbf{k})$.
We therefore focus on \(\rho_{\ell}(\mathbf{i})\) and \(\rho_{\ell}(\mathbf{k})\) as the key quantities governed by $\boldsymbol{\varphi}$. For later use, we also define the projection onto the unit direction $\mathbf{j}$ ($\mathsf{Y}$-axis), obtained by replacing $\mathbf{i}$ or $\mathbf{k}$ with $\mathbf{j}=[0,1,0]^\top$.

Since \(\boldsymbol{a}_{\ell}^{(0)}\) is random under the prior in Assumption~\ref{Assumption of AoA Prior}, the projections \(\rho_\ell(\mathbf{i})\), \(\rho_\ell(\mathbf{k})\), and \(\rho_\ell(\mathbf{j})\) are also random. Under the Gaussian prior
\(\theta_\ell^{(0)} \sim \mathcal{N}(\mu_\ell,\sigma_\ell^2)\) and
\(\phi_\ell^{(0)} \sim \mathcal{N}(\xi_\ell,\varsigma_\ell^2)\), with
\(\theta_\ell^{(0)} \bot \, \phi_\ell^{(0)}\), we define the first and second moments of these projections as
\begin{align}\label{eq:moments}
\bar\rho_{\ell}^{(1)} &= \mathsf{E}[\rho_{\ell}(\mathbf{i})], \quad
\bar\rho_{\ell}^{(2)} = \mathsf{E}[\rho_{\ell}(\mathbf{k})], \quad
\bar\rho_{\ell}^{(3)} = \mathsf{E}[\rho_{\ell}(\mathbf{j})], \nonumber\\
\nu_{\ell}^{(1)} &= \mathsf{E}[\rho_{\ell}(\mathbf{i})^2], \ \
\nu_{\ell}^{(2)} = \mathsf{E}[\rho_{\ell}(\mathbf{k})^2], \ 
\nu_{\ell}^{(3)} = \mathsf{E}[\rho_{\ell}(\mathbf{j})^2].
\end{align}
Because \(\theta_\ell^{(0)}\) and \(\phi_\ell^{(0)}\) are independent Gaussian random variables, these moments admit closed-form expressions via their characteristic functions, given in Appendix~\ref{app:moments}.

Using these moments, we design the objective and constraints of the optimization problem introduced in the~sequel. 

\subsubsection{Objective Function}
The projections of different SPs should be well separated on average. This can be promoted by enlarging the mean separation between every SP pair $(\ell, u)$, as quantified by $|\bar\rho_{\ell}^{(1)}-\bar\rho_{u}^{(1)}|$ and $|\bar\rho_{\ell}^{(2)}-\bar\rho_{u}^{(2)}|$. Accordingly, we define the following objective function:
\begin{align}\label{eq: objective function}
f\left(\boldsymbol{\varphi}\right)=\sum_{\forall \ell, u} (\bar\rho_{\ell}^{(1)}-\bar\rho_{u}^{(1)})^2+(\bar\rho_{\ell}^{(2)}-\bar\rho_{u}^{(2)})^2.
\end{align}
\vspace{-5mm}
\subsubsection{Order Reversal Constraint}
Maximizing \(f(\boldsymbol{\varphi})\) does not explicitly guarantee reliable separation for every SP pair. Even when the objective value is large, some SP pairs may still remain highly overlapped. Furthermore, 
certain orientations may enlarge the mean separations while simultaneously increasing the variability of the SP projections. Their spread can then become comparable to, or even larger than, the mean gap to neighboring SPs, so overlap may still occur in individual realizations. In such cases, the realized ordering of two SP projections may even become opposite to their mean ordering, which we refer to as an \emph{order-reversal} event, as shown in Fig.~\ref{fig:epsilon_separation} (a). To prevent it, we make the following stochastic constraint:

\begin{definition}[$\epsilon$-Separation]
\emph{Two SPs $\ell$ and $u$ are said to be $\epsilon$-separated if
\begin{align}\label{eq:epsilonvalue}
\mathsf{Pr}\!\left[
\big(\rho_{\ell}(\mathbf{i})-\rho_{u}(\mathbf{i})\big)
\big(\bar\rho_{\ell}^{(1)}-\bar\rho_{u}^{(1)}\big)\le 0
\right]\le \epsilon^{(1)},  
\nonumber\\
\mathsf{Pr}\!\left[
\big(\rho_{\ell}(\mathbf{k})-\rho_{u}(\mathbf{k})\big)
\big(\bar\rho_{\ell}^{(2)}-\bar\rho_{u}^{(2)}\big)\le 0
\right]\le \epsilon^{(2)}, 
\end{align}
where \(\epsilon^{(1)}(\ll 1)\) and \(\epsilon^{(2)}(\ll 1)\) denote the allowable probabilities of order reversal on the $\mathsf{X}$- and $\mathsf{Z}$-axes, respectively.} 
\end{definition}\vspace{-1mm}

Together with the mean-separation objective \(f(\boldsymbol{\varphi})\) in \eqref{eq: objective function}, this \(\epsilon\)-separation constraint promotes large mean gaps while controlling the dispersion of each projection, thereby preventing excessive overlap (see Fig.~\ref{fig:epsilon_separation}\,(b)). However, the stochastic nature of \eqref{eq:epsilonvalue} makes the exact closed-form characterization of these constraints intractable. We thus consider the following sufficient condition as a relaxation.
\vspace{-1mm}
\begin{proposition}[Sufficient Condition for $\epsilon$-Separation]\emph{A sufficient condition for all SPs to satisfy \(\epsilon\)-separation is 
\begin{align}\label{eq:epsilon-separation sufficient condition}\tag{C1}
c_{\ell,u}^{(1)}(\boldsymbol{\varphi})\ge 0,\quad c_{\ell,u}^{(2)}(\boldsymbol{\varphi})\ge 0, \quad \forall \ell,u,
\end{align}
where 
\begin{align}
c_{\ell,u}^{(n)}(\boldsymbol{\varphi})
&\triangleq
\left(\bar\rho_{\ell}^{(n)}-\bar\rho_{u}^{(n)}\right)^2 \nonumber\\
&-\!\frac{1\!-\!\epsilon^{(n)}}{\epsilon^{(n)}}
\left(\nu_{\ell}^{(n)}\!+\!\nu_{u}^{(n)}
\!-\!\left(\bar\rho_{\ell}^{(n)}\right)^2
\!-\!\left(\bar\rho_{u}^{(n)}\right)^2
\right),
\end{align}
with \(n\in\{1,2\}\).}
\end{proposition}
\begin{IEEEproof}
See Appendix~\ref{app:moments2}.
\end{IEEEproof}
\begin{figure}[t]
  \centering
  \includegraphics[width=0.9\linewidth]{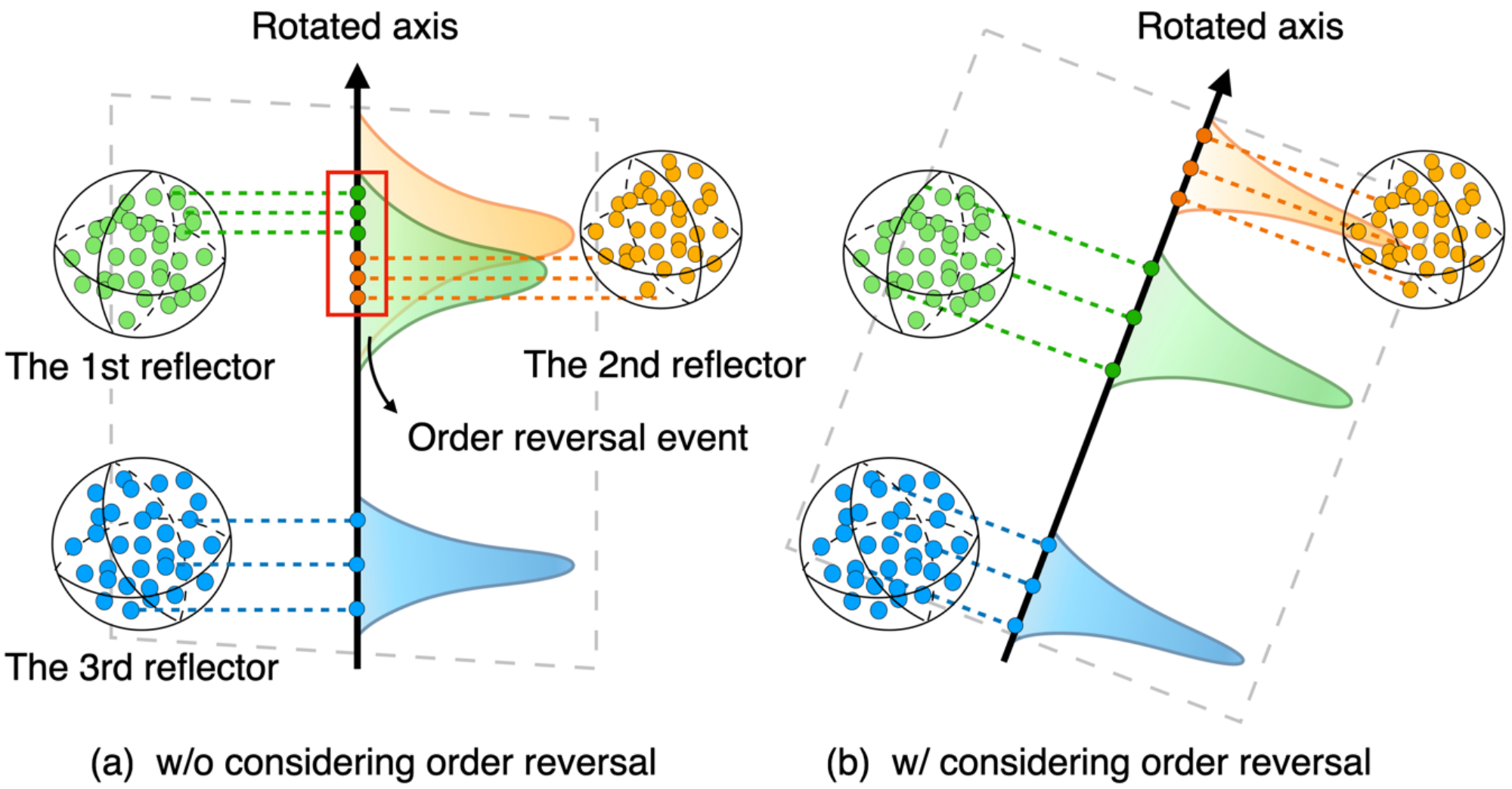}
    \vspace{-3mm}
\caption{Graphical example of the effect of order-reversal constraint. }
\vspace{-5mm}
  \label{fig:epsilon_separation}
\end{figure}

\subsubsection{Front-Side Constraint}
Since the MA is mounted on only one side of the movable plate, each SP should arrive from the front side of the plate to ensure reliable reception. 
The \(\ell\)-th SP is thus said to satisfy the \emph{front-side condition} if
\begin{align}
\mathsf{Pr}\!\left[\rho_\ell(\mathbf{j})\le 0\right]\le \epsilon^{(3)},
\label{eq:front_prob}
\end{align}
where \(\epsilon^{(3)}(\ll\!1)\) represents the allowable probability of violating the front-side condition. 
Similar to the order-reversal constraint, a tractable sufficient condition for the stochastic constraint above can be derived as follows.
\begin{proposition}[Sufficient Condition for Front-Side Incidence]
\emph{A sufficient condition for all SPs to satisfy the front-side incidence constraint is }
\begin{align}\label{eq:front-side sufficient condition}\tag{C2}
c_{\ell}^{(3)}(\boldsymbol{\varphi})\ge 0,
\end{align}
\emph{where}
\(c_{\ell}^{(3)}(\boldsymbol{\varphi})
\triangleq
\bar\rho_{\ell}^{(3)}
-
\sqrt{
\frac{1-\epsilon^{(3)}}{\epsilon^{(3)}}
\left(
\nu_{\ell}^{(3)}-\left(\bar\rho_{\ell}^{(3)}\right)^2
\right)
}.\)

\end{proposition}
\begin{IEEEproof}
See Appendix~\ref{app:moments3}.
\end{IEEEproof}

Combining the above objective and constraints, the plate-orientation design problem is formulated as
\begin{align}
\max_{\boldsymbol{\varphi}} f(\boldsymbol{\varphi}) 
\quad \text{s.t. }
\eqref{eq:epsilon-separation sufficient condition}, \eqref{eq:front-side sufficient condition}.
\label{eq:problemformulation1}\tag{P1}
\end{align}
Since \ref{eq:problemformulation1} is a non-convex constrained problem, we adopt sequential quadratic programming~(SQP) \cite{SQP2000}, which solves a sequence of local quadratic programs obtained by approximating the objective and constraints. The algorithm converges to a local stationary solution, denoted by \(\boldsymbol{\varphi}^*\).

\section{AoA Estimation under Two Linear Movable Antenna Scans}
Based on the optimized orientation \(\boldsymbol{\varphi}^*\) obtained in the preceding section, this section estimates the elevation and azimuth AoAs for all SPs, denoted by \(\{(\theta_{\ell},\phi_{\ell})\}_{\ell=1}^L\), from the two predetermined linear MA scans.  

\subsection{Spatial-Frequency Parameter Extraction}
We first describe how to extract two types of AoA-related parameters, referred to as \emph{spatial-frequency parameters} (SFPs).  
Specifically, when the MA is located at $\boldsymbol{p}_m$, the additional propagation distance in  \eqref{eq: propagation dist} can be rewritten as 
\begin{align}
\rho_\ell(\boldsymbol{p}_m)
= x_m  {u}_\ell + z_m {v}_\ell,
\label{eq:rho_uv_def}
\end{align}
where \({u}_\ell\triangleq \sin\theta_\ell\cos\phi_\ell\) and
\({v}_\ell\triangleq \cos\theta_\ell\) are the SFPs. Recall that 
the MA collects \(M\) samples along each of the \(\mathsf X\)- and \(\mathsf Z\)-axis scans. The parameters \({u}_\ell\) and \({v}_\ell\) are then estimated separately from the measurements obtained along the \(\mathsf X\)- and \(\mathsf Z\)-axes using MUSIC \cite{Schmidt1986} (See Appendix~\ref{app:moments4}). This yields the two SFP sets
\begin{align}
\mathcal{U} = \{ \hat{{u}}^{(1)}, \dots,  \hat{{u}}^{(L)}\}, \quad
\mathcal{V} = \{ \hat{{v}}^{(1)}, \dots,  \hat{{v}}^{(L)}\},
\end{align}
where \(L\) denotes the number of SPs. Since AoA recovery requires one estimate from each axis, the number of resolvable SPs is determined by the smaller number of detected SFPs on the two axes. Here, we assume this number to be equal to the number of detected SFPs. The indices of the detected SFPs in \(\mathcal{U}\) and \(\mathcal{V}\) simply label the extracted components and do not correspond to the SP indices. Therefore, they should be properly paired to recover the elevation and azimuth AoAs, as discussed in the following subsection.

\subsection{Spatial-Frequency Parameter Pairing}
Since the estimated SFPs in \(\mathcal V\) and \(\mathcal U\) are unordered, each set of \(L\) detected SFPs has \(L!\) possible orderings. Let \(\mathcal S\) denote the set of all such orderings. We define the optimal permutation pair as \(\Theta^*,\Xi^*\in\mathcal S\), which specifies the correct ordering of the SFPs in \(\mathcal U\) and \(\mathcal V\) corresponding to the same SP. Under this optimal pairing, the SFP pairs of all SPs are given by \(\{(\hat{\mathrm{u}}^{(\Theta^*_\ell)},\hat{\mathrm{v}}^{(\Xi^*_\ell)})\}_{\ell=1}^L\), where \(\Theta^{*}_\ell\) and \(\Xi^{*}_\ell\) denote the indices of the selected SFPs from \(\mathcal U\) and \(\mathcal V\), respectively.

For correct permutation pairing, we leverage two clues: the first is the sequence of received signals obtained along two linear MA scans in \eqref{eq:scan}, while the second is the prior AoA statistics in Assumption~\ref{Assumption of AoA Prior}, whose details are given as follows.

\subsubsection{Received-Signal Log-Likelihood}
For a given permutation pair \((\Theta,\Xi)\) and the resultant SFP pairs \(\{(\hat{{u}}^{(\Theta_\ell)},\hat{{v}}^{(\Xi_\ell)})\}_{\ell=1}^L\), we construct the steering matrix to model the received signal across the \(2M\) positions as
\begin{align}
\mathbf\Upsilon^{(k)}(\Theta,\Xi)
=
\big[\boldsymbol\Phi^{(k)}_1(\Theta,\Xi),\dots,\boldsymbol\Phi^{(k)}_L(\Theta,\Xi)\big],
\end{align}
where \(\boldsymbol\Phi_\ell^{(k)}(\Theta,\Xi)\in\mathbb C^{2M\times 1}\) denotes the steering vector of the \(\ell\)-th candidate SP. 
Its \(m\)-th element at \(\boldsymbol p_m\) in \eqref{eq:rho_uv_def} is given by
\begin{align}
\big[\boldsymbol\Phi_\ell^{(k)}(\Theta,\Xi)\big]_m
\!=\!
\exp\!\left(
-j2\pi f_k
\frac{x_m \hat{{u}}^{(\Theta_\ell)} + z_m \hat{{v}}^{(\Xi_\ell)}}{c}
\right),
\end{align}
which represents the phase response of the \(\ell\)-th SP conditioned on \((\Theta,\Xi)\). We then stack the received signals on the \(k\)-th subcarrier into
\(
\boldsymbol s^{(k)}=[s^{(k)}(\boldsymbol p_1),\dots,s^{(k)}(\boldsymbol p_{2M})]^\top.
\)
Under the Gaussian model in \eqref{eq:received_signal}, the log-likelihood of the received signal is
\begin{align}
\mathcal L(\Theta,\Xi)
&\triangleq
\log {p}\!\left(\{\boldsymbol s^{(k)}\}_{k=1}^{K}\mid \Theta,\Xi,\{{\hat{\boldsymbol{q}}}^{(k)}\}_{k=1}^{K}\right)
\nonumber\\
&=
-\frac{\sum_{k=1}^{K}
\left\|
\boldsymbol s^{(k)}-\boldsymbol\Upsilon^{(k)}(\Theta,\Xi) \hat{\boldsymbol{q}}^{(k)}
\right\|_2^2}{PN_0/K}
+\mathcal{C}_1,
\label{eq:L_data}
\end{align}
where \(\mathcal{C}_1=-2KM\log\!\left(\pi PN_0/K\right)\). The vector \(\boldsymbol {\hat{q}}^{(k)}=[\hat{q}_1^{(k)},\dots,\hat{q}_L^{(k)}]^\top\in\mathbb C^{L\times 1}\) contains the complex path coefficients  associated with the \(L\) candidate SPs on the \(k\)-th subcarrier, which can be obtained through the least-squares estimate while omitting the details due to the page limit. 

\subsubsection{Prior Statistics}
Since the AoAs in Assumption~\ref{Assumption of AoA Prior} are defined in the initial coordinate system \((\mathsf{X}^{(0)}, \mathsf{Y}^{(0)}, \mathsf{Z}^{(0)})\), the estimated SFP pairs \(\{(\hat{{u}}^{(\Theta_\ell)},\hat{{v}}^{(\Xi_\ell)})\}_{\ell=1}^L\), which are defined on the rotated coordinate system \((\mathsf{X}, \mathsf{Y}, \mathsf{Z})\), should be transformed into the initial one by multiplying \(\boldsymbol R(\boldsymbol{\varphi}^\star)\) to the corresponding unit direction vector. 
Let \((\hat{\theta}^{(0)}_\ell,\hat{\phi}^{(0)}_\ell)\) denote the resultant AoA pair defined in the initial coordinate system.
Under the Gaussian model in Assumption~\ref{Assumption of AoA Prior}, the prior probability function is
\begin{align}
\mathcal{D}(\Theta,\Xi)&\triangleq{p}\!\left(\{\hat{\theta}_\ell^{(0)}\!,\!\hat{\phi}_\ell^{(0)}\}_{\ell=1}^{L} \!\mid\! \Theta,\Xi\right)
\nonumber \\
&=\prod_{\ell=1}^{L}
\frac{\text{exp} \left( {
-\frac{\big(\hat{\theta}^{(0)}_\ell\!-\!\mu_\ell\big)\!^2}{2\sigma_\ell^2}
-\frac{\big(\hat{\phi}^{(0)}_\ell\!-\!\xi_\ell\big)\!^2}{2\varsigma_\ell^2}
}\right)}{
2\pi \sigma_\ell \varsigma_\ell
}.\label{eq:D_data}
\end{align}

By combining \eqref{eq:L_data} and \eqref{eq:D_data}, we formulate the MAP-based pairing rule as 
\begin{align}\label{eq:map_matching_rule}
(\Theta^*,\Xi^*)
=
\arg\max_{\Theta,\Xi\in\mathcal S}
\Big(
\mathcal{L}(\Theta,\Xi)
+
\log \mathcal{D}(\Theta,\Xi)
\Big).
\end{align}
A larger \(\mathcal{L}(\Theta,\Xi)\) indicates better consistency with the received signal, while a larger \(\mathcal{D}(\Theta,\Xi)\) indicates better consistency with the prior AoA statistics. 
Based on the resulting matched pairs, the AoA pairs \(\{(\hat\theta_{\ell}^{(0)}, \hat\phi_{\ell}^{(0)})\}_{\ell=1}^{L}\) can then be recovered.

\section{Simulation Results} \vspace{-1.5mm}
In this section, we evaluate the proposed MA-based AoA sensing framework. The carrier frequency is set to \(28\) GHz, the bandwidth to \(50\) MHz, and the number of subcarriers to \(K=64\). The MA scans \(M=32\) positions along each predetermined linear axis, with \(d=\lambda/2\). We set \(L=4\) with prior mean AoAs \((\mu_\ell,\xi_\ell)\) given by \((115^\circ,55^\circ)\), \((100^\circ,115^\circ)\), \((50^\circ,40^\circ)\), and \((50^\circ,120^\circ)\). The setup is chosen so that the received signals along the \(\mathsf{X}\)- and \(\mathsf{Z}\)-axis scans are highly overlapped across SPs. Five cases with \(\sigma_\ell=\varsigma_\ell\in\{2^\circ,4^\circ,6^\circ,8^\circ,10^\circ\}\) are considered, and the AoAs are drawn from the corresponding Gaussian priors in each Monte Carlo trial. The hyperparameters 
\(\epsilon^{(1)},\epsilon^{(2)},\epsilon^{(3)}\) are numerically tuned for each prior spread. All results are obtained from 5000 Monte Carlo trials, and the performance is evaluated by the joint AoA RMSE, defined as
\(
\left(
\frac{1}{L}
\sum_{\ell=1}^{L}
\left[
\big({\theta}_{\ell}^{(0)}-\hat\theta_{\ell}^{(0)}\big)^2
+
\big({\phi}_{\ell}^{(0)}-\hat\phi_{\ell}^{(0)}\big)^2
\right]
\right)^{\frac{1}{2}}\). 
\begin{figure}[t]
  \centering
  \includegraphics[width=0.9\linewidth,height=0.6\linewidth,keepaspectratio=false]{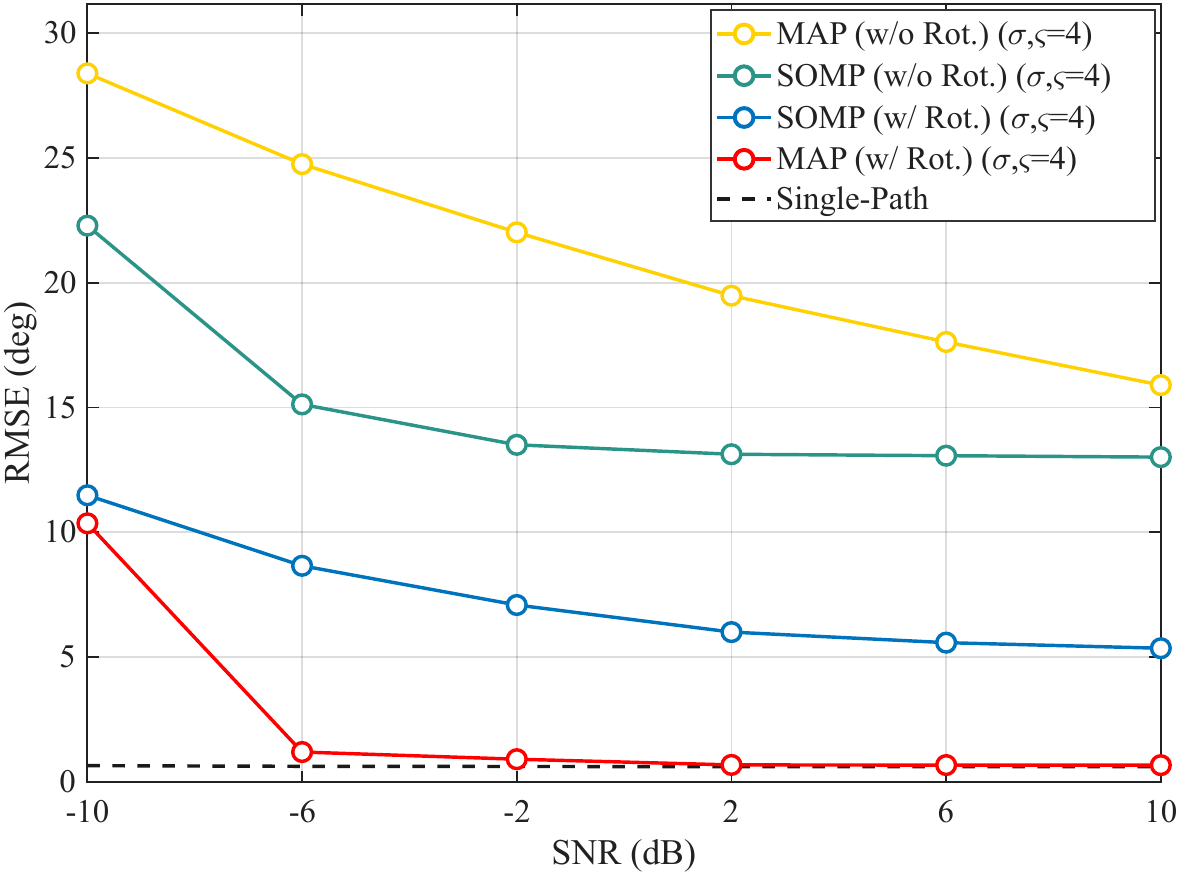}
  \vspace{-3mm}
  \caption{Joint AoA RMSE versus SNR}
  \vspace{-6mm}
  \label{fig:snr}
\end{figure}

For AoA estimation, the proposed {MAP (w/ Rot.)} is compared with four benchmarks: {MAP (w/o Rot.)}, which performs MAP-based pairing in the initial coordinate system, {SOMP (w/o Rot.)}, which performs SOMP-based joint pairing \cite{SCao2025} without rotation, {SOMP (w/ Rot.)}, which applies the SOMP-based pairing in the optimized rotated coordinate system, and {Single-Path}, which serves as a reference case.

Fig.~\ref{fig:snr} shows the joint AoA RMSE versus SNR for \(\sigma=4^\circ\). The rotation-based methods outperform their non-rotated counterparts, confirming the benefit of the optimized plate orientation in improving SP separability. In particular, {MAP (w/ Rot.)} consistently outperforms {SOMP (w/ Rot.)}, since it exploits both the received signals and the prior AoA statistics in the pairing process. As SNR increases, the performance of {MAP (w/ Rot.)} approaches that of {Single-Path}, indicating near-optimal performance in the high-SNR regime.

Fig.~\ref{fig:variance} shows the joint AoA RMSE of {MAP (w/ Rot.)} versus \(\sigma\) at SNR \(=10\) dB. The proposed method remains close to the {Single-Path} reference when \(\sigma\) is small. As \(\sigma\) increases, its performance degrades because the prior becomes less informative and the larger uncertainty makes both MUSIC-based SFP estimation and the subsequent MAP pairing less reliable. Even in this regime, {MAP (w/ Rot.)} still outperforms the benchmark methods, demonstrating its robustness.
\vspace{-2mm}
\label{eq:joint_rmse}
\begin{figure}[t]
  \centering
  \includegraphics[width=0.9\linewidth,height=0.6\linewidth,keepaspectratio=false]{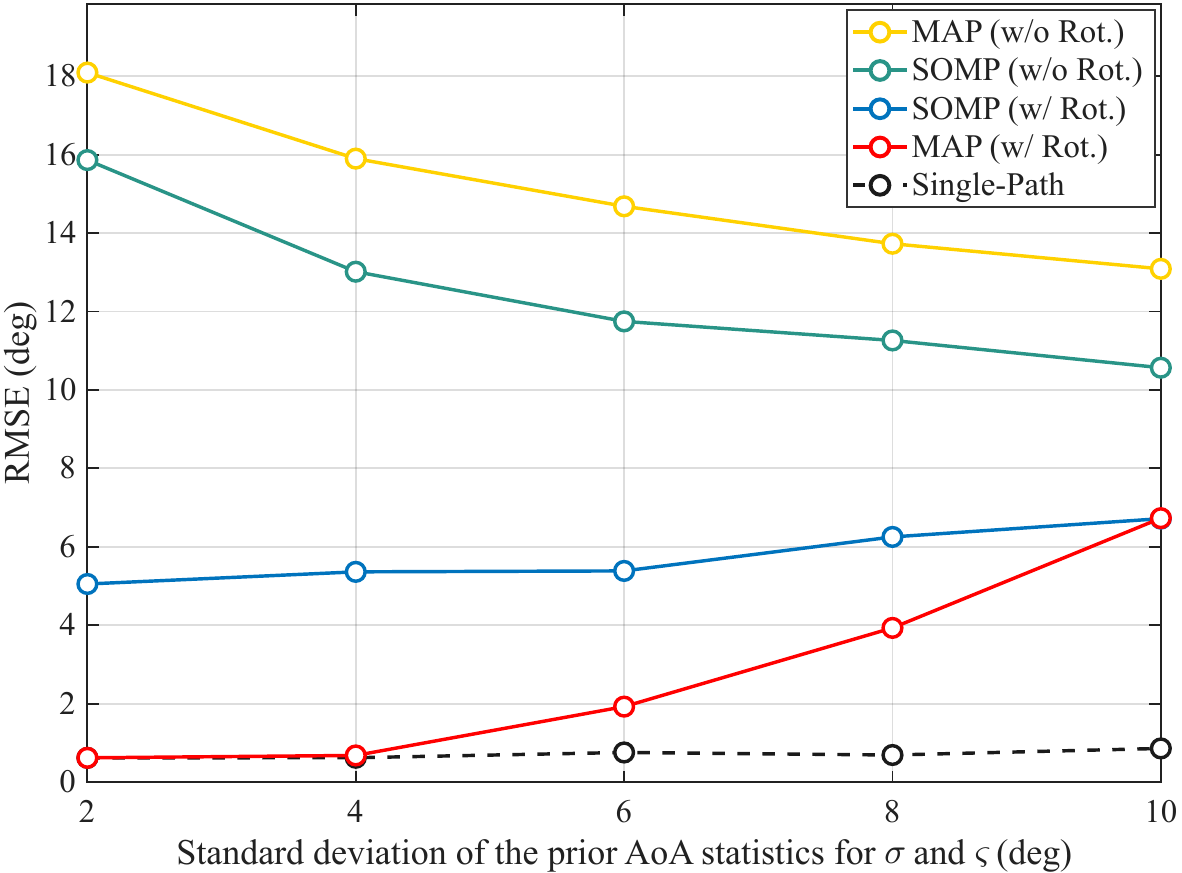}
  \vspace{-3mm}
  \caption{Joint AoA RMSE for different standard deviations of the prior AoA statistics.}
  \vspace{-6mm}
  \label{fig:variance}
\end{figure}

\section{Conclusion}
This paper proposes a prior-aided MA control framework for multi-path sensing. By optimizing the movable-plate orientation and two linear MA scans, the proposed scheme exploits coarse prior AoA statistics to improve the separability of closely spaced SPs and recover the AoAs of all paths via MAP-based pairing. Future work will extend the framework to joint ToA and AoA estimation, as well as broader sensing tasks such as environment mapping and localization.
  \vspace{-2mm}
\bibliographystyle{IEEEtran}
\bibliography{reference/ieeeabrv,reference/reference}

\appendices
\section{Moment Definitions and Derivations}
\label{app:moments}

\subsection{Derivation of the First Moment}
Since both $\theta_\ell^{(0)}$ and $\phi_\ell^{(0)}$ follow Gaussian distributions, we utilize the characteristic function of a generic Gaussian random variable $X \sim \mathcal{N}(\mu, \sigma^2)$, given as 
\begin{align}
\mathsf{E}[\cos X] &= \cos(\mu)e^{-\sigma^2/2}, \nonumber\\
\mathsf{E}[\sin X] &= \sin(\mu)e^{-\sigma^2/2}.
\end{align}

The first moment along the \(\mathbf{i}\)-axis is
\begin{align}
\bar\rho_{\ell}^{(1)}
&=
\mathsf{E}\!\left[ [\boldsymbol{R}^\top(\boldsymbol{\varphi})\boldsymbol{a}^{(0)}_\ell]_1 \right] \nonumber\\
&=
\cos\beta\, \mathsf{E}[\sin\theta_\ell^{(0)}]\,
\mathsf{E}[\cos(\phi_\ell^{(0)}-\alpha)]
-\sin\beta\,\mathsf{E}[\cos \theta_\ell^{(0)}].
\end{align}

Similarly, the first moment along the \(\mathbf{k}\)-axis is
\begin{align}
\bar\rho_{\ell}^{(2)}
&=
\mathsf{E}\!\left[ [\boldsymbol{R}^{\top}(\boldsymbol{\varphi})\boldsymbol{a}^{(0)}_\ell]_3 \right] \nonumber\\
&=
\mathsf{E}[\sin \theta_\ell^{(0)}]
\Big(
(\cos \alpha \sin \beta \cos \gamma + \sin \alpha \sin \gamma)\mathsf{E}[\cos \phi_\ell^{(0)}] \nonumber\\
&\quad
+(\sin \alpha \sin \beta \cos \gamma - \cos \alpha \sin \gamma)\mathsf{E}[\sin \phi_\ell^{(0)}]
\Big) \nonumber\\
&\quad
+(\cos \beta \cos \gamma)\mathsf{E}[\cos \theta_\ell^{(0)}].
\end{align}

Lastly, the first moment along the \(\mathbf{j}\)-axis is
\begin{align}
\bar\rho_{\ell}^{(3)}
&=
\mathsf{E}\!\left[ [\boldsymbol{R}^{\top}(\boldsymbol{\varphi})\boldsymbol{a}^{(0)}_\ell]_2 \right] \nonumber\\
&=
\mathsf{E}[\sin \theta_\ell^{(0)}]
\Big(
(\cos \alpha \sin \beta \sin \gamma - \sin \alpha \cos \gamma)\mathsf{E}[\cos \phi_\ell^{(0)}] \nonumber\\
&\quad
+(\sin \alpha \sin \beta \sin \gamma + \cos \alpha \cos \gamma)\mathsf{E}[\sin \phi_\ell^{(0)}]
\Big) \nonumber\\
&\quad
+(\cos \beta \sin \gamma)\mathsf{E}[\cos \theta_\ell^{(0)}].
\end{align}

\subsection{Derivation of the Second Moment}
Utilizing the characteristic function, the expectations of the squared trigonometric terms are given as
\begin{align} \label{eq: expectations}
\mathsf{E}[\cos^2 \theta_\ell^{(0)}]
&= \frac{1}{2}\bigl(1 + e^{-2\sigma_\ell^2}\cos 2\mu_\ell\bigr), \nonumber\\
\mathsf{E}[\sin^2 \theta_\ell^{(0)}]
&= \frac{1}{2}\bigl(1 - e^{-2\sigma_\ell^2}\cos 2\mu_\ell\bigr), \nonumber\\
\mathsf{E}[\sin \theta_\ell^{(0)} \cos \theta_\ell^{(0)}]
&= \frac{1}{2} e^{-2\sigma_\ell^2}\sin 2\mu_\ell.
\end{align}
Analogous expressions hold for \(\phi_\ell^{(0)}\) with parameters \((\xi_\ell,\varsigma_\ell)\).

For the second moment along the $\mathbf{i}$-axis, expanding the squared term yields
\begin{align}
\nu_\ell^{(1)}
&=
\mathsf{E}\!\left[ \left([\boldsymbol{R}^\top(\boldsymbol{\varphi})\boldsymbol{a}^{(0)}_\ell]_1\right)^2 \right]\nonumber\\
&=
\cos^2 \beta \,\mathsf{E}[\sin^2 \theta_\ell^{(0)}]\,
\mathsf{E}\!\big[\cos^2(\phi_\ell^{(0)}-\alpha)\big] \nonumber\\
&\quad
+\sin^2\beta\, \mathsf{E}[\cos^2\theta_\ell^{(0)}] \nonumber\\
&\quad
-\sin 2\beta\,
\mathsf{E}[\sin \theta_\ell^{(0)}\cos\theta_\ell^{(0)}]\,
\mathsf{E}\!\big[\cos(\phi_\ell^{(0)}-\alpha)\big].
\end{align}
Similarly, the second moment along the \(\mathbf{k}\)-axis is
\begin{align}
\nu_\ell^{(2)}
&=
\mathsf{E}\!\left[ \left([\boldsymbol{R}^\top(\boldsymbol{\varphi})\boldsymbol{a}^{(0)}_\ell]_3\right)^2 \right]\nonumber\\
&=
\mathsf{E}[\sin^2 \theta_\ell^{(0)}]
\Big(
\mathcal{Q}(\boldsymbol\varphi)^2 \mathsf{E}[\cos^2 \phi_\ell^{(0)}]
+\mathcal{H}(\boldsymbol\varphi)^2 \mathsf{E}[\sin^2 \phi_\ell^{(0)}] \nonumber\\
&\quad
+2\mathcal{Q}(\boldsymbol\varphi)\mathcal{H}(\boldsymbol\varphi)
\mathsf{E}[\sin \phi_\ell^{(0)} \cos \phi_\ell^{(0)}]
\Big) \nonumber\\
&\quad
+(\cos \beta \cos \gamma)^2 \mathsf{E}[\cos^2 \theta_\ell^{(0)}] \nonumber\\
&\quad
+2(\cos \beta \cos \gamma)\mathsf{E}[\sin \theta_\ell^{(0)} \cos \theta_\ell^{(0)}] \nonumber\\
&\quad \times
\Big(
\mathcal{H}(\boldsymbol\varphi)\mathsf{E}[\sin \phi_\ell^{(0)}]
+\mathcal{Q}(\boldsymbol\varphi)\mathsf{E}[\cos \phi_\ell^{(0)}]
\Big), 
\end{align}
where $\mathcal{Q}(\boldsymbol\varphi)= \cos \alpha\sin \beta \cos \gamma + \sin \alpha \sin \gamma$ and $\mathcal{H}(\boldsymbol\varphi)= \sin \alpha \sin \beta\cos \gamma - \cos \alpha\sin \gamma.$

Lastly, the second moment along the \(\mathbf{j}\)-axis is 
\begin{align}
\nu_\ell^{(3)}
&=
\mathsf{E}\!\left[ \left(
[\boldsymbol{R}^\top(\boldsymbol{\varphi})\boldsymbol{a}^{(0)}_\ell]_2
\right)^2 \right]\nonumber\\
&=
\mathsf{E}[\sin^2 \theta_\ell^{(0)}]
\Big(
\mathcal{P}(\boldsymbol\varphi)^2 \mathsf{E}[\cos^2 \phi_\ell^{(0)}]
+\mathcal{G}(\boldsymbol\varphi)^2 \mathsf{E}[\sin^2 \phi_\ell^{(0)}] \nonumber\\
&\quad
+2\mathcal{P}(\boldsymbol\varphi)\mathcal{G}(\boldsymbol\varphi)
\mathsf{E}[\sin \phi_\ell^{(0)} \cos \phi_\ell^{(0)}]
\Big) \nonumber\\
&\quad
+(\cos \beta \sin \gamma)^2 \mathsf{E}[\cos^2 \theta_\ell^{(0)}] \nonumber\\
&\quad
+2(\cos \beta \sin \gamma)\mathsf{E}[\sin \theta_\ell^{(0)} \cos \theta_\ell^{(0)}] \nonumber   \\
&\quad \times
\Big(
\mathcal{P}(\boldsymbol\varphi)\mathsf{E}[\cos \phi_\ell^{(0)}]
+\mathcal{G}(\boldsymbol\varphi)\mathsf{E}[\sin \phi_\ell^{(0)}]
\Big),
\end{align}
where
$\mathcal{P}(\boldsymbol\varphi) = \cos \alpha \sin \beta \sin \gamma - \sin \alpha \cos \gamma$ and
$\mathcal{G}(\boldsymbol\varphi) = \sin \alpha \sin \beta \sin \gamma + \cos \alpha \cos \gamma.$

\section{Sufficient Conditions for \texorpdfstring{$\epsilon$}{epsilon}-Separation and Front-Side Incidence}
\subsection{Derivation of the \texorpdfstring{$\epsilon$}{epsilon}-Separation }\label{app:moments2}
We focus on proving $\epsilon^{(1)}$-separation on the $\mathsf{X}$-axis, while the proof of $\epsilon^{(2)}$-separation on the $\mathsf{Z}$-axis follows by the same argument after replacing $\mathbf{i}$ with $\mathbf{k}$ and ${(1)}$ with ${(2)}$.
Given $\bar\rho_{\ell}^{(1)}\geq \bar\rho_{u}^{(1)}$, define the difference random variable $V_{\ell,u}^{(1)}=\rho_{\ell}^{(1)}-\rho_{u}^{(1)}$ whose mean and variance are 
\begin{align}
\mathsf{E}[V_{\ell,u}^{(1)}]&=\bar\rho_{\ell}^{(1)}-\bar\rho_{u}^{(1)},\label{Proof_Proposition1_1}\\
\mathsf{var}[V_{\ell,u}^{(1)}]&=\mathsf{var}[\rho_{\ell}^{(1)}]+\mathsf{var}[\rho_{u}^{(1)}]\nonumber\\
&=\nu_{\ell}^{(1)}-(\bar\rho_{\ell}^{(1)})^2+\nu_{u}^{(1)}-(\bar\rho_{u}^{(1)})^2.\label{Proof_Proposition1_2}
\end{align}
The order-reversal event can be written as
\begin{align}
\mathsf{Pr}[\rho_{\ell}^{(1)}< \rho_{u}^{(1)}]=\mathsf{Pr}[V_{\ell,u}^{(1)}<0]\overset{(a)}{\leq} \frac{\mathsf{var}[V_{\ell,u}^{(1)}]}{\mathsf{var}[V_{\ell,u}^{(1)}]+(\mathsf{E}[V_{\ell,u}^{(1)}])^2},\nonumber
\end{align}
where (a) follows from the Cantelli bound. Therefore,  a sufficient condition for $\mathsf{Pr}[V_{\ell,u}^{(1)}\leq 0]\le \epsilon^{(1)}$ is
\begin{align}\label{Proof_Proposition1_3}
\frac{\mathsf{var}[V_{\ell,u}^{(1)}]}{\mathsf{var}[V_{\ell,u}^{(1)}]+(\mathsf{E}[V_{\ell,u}^{(1)}])^2}\leq \epsilon^{(1)},
\end{align}
which can be converted to $c_{\ell,u}^{(1)}(\boldsymbol{\varphi})\geq0$ by plugging \eqref{Proof_Proposition1_1} and \eqref{Proof_Proposition1_2} into \eqref{Proof_Proposition1_3}.

\subsection{Derivation of the Front-Side Incidence}\label{app:moments3}
Define the random variable \(V_{\ell}^{(3)}=\rho_{\ell}^{(3)}\), whose mean and variance are
\begin{align}
\mathsf{E}[V_{\ell}^{(3)}]
&=\bar\rho_{\ell}^{(3)}, \label{Proof_Proposition1_4}\\
\mathsf{var}[V_{\ell}^{(3)}]
&=\mathsf{var}[\rho_{\ell}^{(3)}] =\nu_{\ell}^{(3)}-\left(\bar\rho_{\ell}^{(3)}\right)^2. \label{Proof_Proposition1_5}
\end{align}
The front-side condition can be written as 
\begin{align}
\mathsf{Pr}[\rho_{\ell}^{(3)}<0]
&=\mathsf{Pr}[V_{\ell}^{(3)}<0] \overset{(b)}{\leq}
\frac{\mathsf{var}[V_{\ell}^{(3)}]}
{\mathsf{var}[V_{\ell}^{(3)}]+\left(\mathsf{E}[V_{\ell}^{(3)}]\right)^2}, \nonumber
\end{align}
where \((b)\) follows from the Cantelli bound. Therefore, a sufficient condition for
$\mathsf{Pr}[V_{\ell}^{(3)}<0]\le \epsilon^{(3)}$ is
\begin{align}\label{Proof_Proposition1_6}
\frac{\mathsf{var}[V_{\ell}^{(3)}]}
{\mathsf{var}[V_{\ell}^{(3)}]+\left(\mathsf{E}[V_{\ell}^{(3)}]\right)^2}
\le \epsilon^{(3)},
\end{align}
which can be converted to $c^{(3)}_{\ell}(\boldsymbol{\varphi}) \ge 0 $ by plugging 
\eqref{Proof_Proposition1_4} and \eqref{Proof_Proposition1_5} into \eqref{Proof_Proposition1_6}.

\section{MUSIC Algorithm for AoA Estimation}\label{app:moments4}
\label{app:MUSIC_algorithm}
We describe the MUSIC algorithm with spatial smoothing. Let \(\{\boldsymbol p_m^{(1)}\}_{m=1}^M\) and \(\{\boldsymbol p_m^{(2)}\}_{m=1}^M\) denote the \(M\) uniformly spaced measurement positions along the \(\mathsf{X}\)- and \(\mathsf{Z}\)-axes, respectively. Taking the \(\mathsf{X}\)-axis as an example, the array is partitioned into \(Q=M-M_{\mathrm{sub}}+1\) overlapping subarrays, where \(M_{\mathrm{sub}}\) satisfies \(L<M_{\mathrm{sub}}\le M-L+1\). For the \(k\)-th subcarrier, the signal vector of the \(q\)-th subarray, \(q=1,\dots,Q\), is
\begin{equation}
\boldsymbol{s}_{q}^{(k)}
=
[s^{(k)}(\boldsymbol{p}_q^{(1)}), s^{(k)}(\boldsymbol{p}_{q+1}^{(1)}), \dots, s^{(k)}(\boldsymbol{p}_{q+M_{\mathrm{sub}}-1}^{(1)})]^\top,
\end{equation}
where \(s^{(k)}(\boldsymbol p_m^{(1)})\) denotes the demodulated received signal as defined in \eqref{eq:received_signal}. By averaging over $K$ subcarriers and $Q$ subarrays, the rank-restored covariance matrix is obtained as
\begin{equation}
\boldsymbol{C}
=
\frac{1}{KQ}\sum_{k=1}^{K}\sum_{q=1}^{Q}\boldsymbol{s}_{q}^{(k)}(\boldsymbol{s}_{q}^{(k)})^H.
\end{equation}
The eigen-decomposition of $\boldsymbol{C}$ yields $\boldsymbol{C} = \boldsymbol{U}_{\mathsf{s}} \boldsymbol{\Lambda}_{\mathsf{s}} \boldsymbol{U}_{\mathsf{s}}^H + \boldsymbol{U}_{\mathsf{n}} \boldsymbol{\Lambda}_{\mathsf{n}} \boldsymbol{U}_{\mathsf{n}}^H$, where $\boldsymbol{U}_{\mathsf{n}}$ spans the noise subspace. The parameters $\{\hat{v}_\ell\}$ are then estimated by identifying the peaks of the MUSIC pseudo-spectrum:
\begin{equation}\label{eq:music_spectrum_final}
S(\vartheta)
=
\frac{1}{\boldsymbol d^H(\vartheta)\boldsymbol U_{\mathsf n}\boldsymbol U_{\mathsf n}^H\boldsymbol d(\vartheta)},
\end{equation}
where \(\boldsymbol d(\vartheta)=[1,e^{-j\frac{2\pi d}{\lambda}\vartheta},\dots,e^{-j\frac{2\pi d}{\lambda}(M_{\mathrm{sub}}-1)\vartheta}]^\top\) is the steering vector parameterized by $\vartheta \in [-1, 1]$.
The same procedure applies to the \(\mathsf{Z}\)-axis measurements.
\end{document}